\documentclass[conference]{IEEEtran}

\usepackage{graphicx} 
\usepackage{color}
\usepackage{psfrag}
\usepackage{epsfig}
\usepackage{xspace,colortbl}

\definecolor{Red}{rgb}{1,0,0}
\definecolor{Blue}{rgb}{0,0,1}
\definecolor{Olive}{rgb}{0.41,0.55,0.13}
\definecolor{Green}{rgb}{0,1,0}
\definecolor{MGreen}{rgb}{0,0.8,0}
\definecolor{DGreen}{rgb}{0,0.55,0}
\definecolor{Yellow}{rgb}{1,1,0}
\definecolor{Cyan}{rgb}{0,1,1}
\definecolor{Magenta}{rgb}{1,0,1}
\definecolor{Orange}{rgb}{1,.5,0}
\definecolor{Violet}{rgb}{.5,0,.5}
\definecolor{Purple}{rgb}{.75,0,.25}
\definecolor{Brown}{rgb}{.75,.5,.25}
\definecolor{Grey}{rgb}{.5,.5,.5}

\usepackage{booktabs} 
\usepackage{array} 
\usepackage{paralist} 
\usepackage{verbatim} 
\usepackage{subfigure} 
\usepackage{amsmath,amssymb}

\newtheorem{theorem}{Theorem}

\newtheorem{lemma}{Lemma}

\newtheorem{remark}{Remark}

\newtheorem{definition}{Definition}

\usepackage{xspace}
\usepackage{mathrsfs}
\usepackage{amsopn}

\newcommand{\pen}{{P_e^{(n)}}}

\newcommand{\Mh}{{\hat{M}}}

\let\P\relax
\DeclareMathOperator\P{P}

\def\textiid{i.i.d.\@\xspace}
\newcommand\iid{\ifmmode\text{ i.i.d. } \else \textiid \fi}

\newcommand{\qedsymbol}{\rule{1.3ex}{1.3ex}}
\newcommand{\qed}{\mbox{}\hfill\qedsymbol}

\begin{document}

\title{The capacity region of a class of broadcast channels with a  sequence of less noisy receivers}

\author{\IEEEauthorblockN{Zizhou Vincent Wang}
\IEEEauthorblockA{Department of Information Engineering\\
The Chinese University of Hong Kong\\
Sha Tin, N.T., Hong Kong\\
Email: zzwang6@ie.cuhk.edu.hk} \and \IEEEauthorblockN{Chandra Nair}
\IEEEauthorblockA{Department of Information Engineering\\
The Chinese University of Hong Kong\\
Sha Tin, N.T., Hong Kong\\
Email: chandra@ie.cuhk.edu.hk} }


\maketitle

\begin{abstract}
The capacity region of a broadcast channel consisting of
$k$-receivers that lie in a less noisy sequence is an open problem,
when $k \geq 3$. We solve this problem for the case $k=3$. Generalizing this result, we prove 
that superposition coding is 
optimal for a class of broadcast channels with a sequence of less
noisy receivers. The key idea is a new information inequality for less noisy receivers, which may be potentially useful in other problems as well.
\end{abstract}


%
\IEEEpeerreviewmaketitle

\section{Introduction}
Consider the problem of reliable communication of $k$ independent messages $M_1,...,M_k$ over a discrete-memoryless broadcast channel (DM-BC), to  $k$-receivers $Y_1, Y_2, \ldots, Y_k$ respectively. A $(2^{nR_1}\times \cdots \times 2^{nR_k}, n)$ code for the DM-BC  consists of: (i) a message set $[1:2^{nR_1}]\times \cdots \times[1:2^{nR_k}]$, (ii) an encoder that assigns a codeword $x^n(m_1,  \ldots, m_k)$ to each message-tuple $(m_1,  \ldots, m_k)$, and (iii) $k$ decoders, decoder $l$ assigns an estimate $\hat{m}_l(y_{l,1}^n) \in [1:2^{nR_l}]$ or an error message $\mathrm{e}$ to each received sequence $y^n_{l,1}$, $1 \leq l \leq k$. We assume that the messages are uniformly distributed over  is uniformly distributed over $[1:2^{nR}]\times \cdots [1:2^{nR_k}]$. The probability of error is defined as $\pen = \P\{\cup_{l=1}^k \Mh_l \neq M\}$.

A rate-tuple $(R_1, \cdots, R_k)$ is said to be achievable if there exists a sequence of $(2^{nR_1}\times \cdots \times 2^{nR_k}, n)$ codes with $\pen \to 0$ as $n\to \infty$. The capacity region is defined as the closure of the union of all achievable rates.


\smallskip
\begin{definition}
A receiver $Y_s$ is said to be less noisy\cite{kom75} than receiver
$Y_t$ if $I(U;Y_s) \geq I(U;Y_t)$ for all $U \to X \to (Y_s,Y_t)$.
\end{definition}

We denote this relationship(partial-order) by $Y_s \succeq Y_t$.

\begin{remark}
\label{re:om}
Observe that this partial-order only depends  on the marginal distributions $p(y_s|x)$ and $p(y_t|x)$.
\end{remark}

\begin{definition}
A $k$-receiver less noisy broadcast channel  is a $k$-receiver discrete memoryless broadcast channel where the receivers  satisfy the partial order $Y_1 \succeq Y_2 \succeq \cdots \succeq Y_k$.
\end{definition}

The capacity region for the $2$-receiver broadcast channel was established (Proposition 3 in \cite{kom75}) to be the union of rate pairs $(R_1,R_2)$ satisfying
\begin{align*}
R_1 &\leq I(X;Y_1|U) \\
R_2 &\leq I(U;Y_2)
\end{align*} 
over all choices of $(U,X)$ such that $U \to X \to (Y_1,Y_2)$ forms a Markov chain. 

The extension of this result to $k$-receivers is open, $k \geq 3$. In this paper we present a simple proof for the case $k=3$.

Further our proof can also be used to provide an alternate (much-simpler) proof for $k=2$, although it must be noted that the original proof provides a strong-converse while ours provides a weak-converse. A modern-day weak converse proof for the 2-receiver case may also be obtained using the outer bounds in \cite{mar79, elg79, nae07}, however each of these uses Csiszar sum lemma which has no natural generalization to three receivers. Instead our proof relies on a information inequality (Lemma \ref{le:ord}) valid for less noisy-receivers which helps us by-pass the use of Csiszar sum lemma. 

Indeed using this lemma one can also obtain the capacity region for a subset  of $k$-receiver less noisy broadcast channel (which contains the 3-receiver less noisy broadcast channel as well). However for clarity of exposition, we shall first establish the result for the 3-receiver less noisy broadcast channel and then present the general result for the class of $k$-receiver less noisy broadcast channel.

\section{Three-receiver less noisy broadcast channel}

The main result of the paper is the following:
\begin{theorem}
The capacity region of a 3-receiver less noisy discrete memoryless broadcast channel is given by the union of rate triples $(R_1, R_2, R_3)$ satisfying:
\begin{align*}
R_1 &\leq I(X;Y_1|V) \\
R_2 & \leq I(V; Y_2|U) \\
R_3 &\leq I(U;Y_3)
\end{align*}
over all choices of $(U,V,X)$ such that $U \to V \to X \to (Y_1,Y_2, Y_3)$ forms a  Markov chain. Further it suffices to consider $|U| \leq |X| + 1, |V| \leq (|X|+1)^2$.
\end{theorem}
\subsection{Achievability} The rate-triples are achievable using superposition coding and jointly typical decoding. The arguments are standard in literature and hence only a minor outline is provided. 

Consider a $(U,V,X)$ such that $U \to V \to X \to (Y_1,Y_2, Y_3)$ forms a  Markov chain. We will show the achievability of any rate-triple satisfying $R_3 < I(U;Y_3), R_2 < I(V;Y_2|U), R_1 < I(X;Y_1|V)$.

The encoding proceeds as follows: 
\begin{itemize}
\item Generate $2^{n{R_3}}$ sequence  $u^n(m_3) \sim \prod_{i=1}^n p_{U}(u_i)$.
\item For each $m_3$, generate $2^{nR_2}$ sequences  $v^n(m_2,m_3)$ distributed according to  $\prod_{i=1}^n p_{V|U}(v_i|u_i)$.
\item Finally for each $(m_2,m_3)$ pair, generate $2^{nR_1}$ $x^n(m_1,m_2,m_3)$ sequences    distributed according to $\prod_{i=1}^n p_{X|V,U}(x_i|v_i,u_i) =  \prod_{i=1}^n p_{X|V}(x_i|v_i)$.
\end{itemize}

Receiver $Y_3$, upon receiving $y_{31}^n$, assigns $\hat{M}_3=m_3$  if there is a unique sequence $u^n(m_3)$ such that the pair $(u^n(m_3),y_{31}^n)$ is jointly typical; otherwise receiver $Y_3$ declares an error. This decoding succeeds with high probability as long as $R_3 < I(U;Y_3)$.

Receiver $Y_2$ performs successive decoding. (This is in general worse than joint decoding, but in this situation successive decoding is enough.) Upon receiving $y_{21}^n$, assigns $\bar{M}_3=m_3$  if there is a unique sequence $u^n(m_3)$ such that the pair $(u^n(m_3),y_{21}^n)$ is jointly typical; otherwise receiver $Y_2$ declares an error. Assuming if finds a unique $u^n(m_3)$ sequence, it then assigns $\hat{M}_2=m_2$  if there is a unique sequence $v^n(m_2,m_3)$ such that the triple $(u^n(m_3), v^n(m_2,m_3),y_{21}^n)$ is jointly typical; otherwise receiver $Y_2$ declares an error. The first step of decoding succeeds with high probability as long as $R_3 < I(U;Y_2)$, but this holds as $I(U;Y_2) \geq I(U;Y_3)$ (since $Y_2$ is a less-noisy receiver than $Y_3$). The second step of decoding succeeds with high probability as long as $R_2 < I(V;Y_2|U)$.  

Similarly, receiver $Y_1$ also performs successive decoding. The three steps of decoding will succeed with high probability as long as the conditions $R_3 < I(U;Y_1), R_2 < I(V;Y_1|U)$, and $R_1 < I(X;Y_1|V,U) = I(X;Y_1|V)$ hold.
Since $Y_1 \succeq Y_2 \succeq Y_3$ the first two conditions are automatically satisfied. This completes the proof of achievability. \qed

\subsection{Converse}
The interesting part of this proof is the converse, and in particular the use of  Lemma \ref{le:ord}  to identify the auxiliary random variables.

\begin{lemma}
\label{le:ord}
Let $Y_s \succeq Y_t$, and $M$ be any random variable such that
$$ M \to X^n \to (Y_{s,1}^n, Y_{t,1}^n) $$
form a Markov chain. Then the following hold:
\begin{enumerate}
\item $ I(Y_{s,1}^{i-1};Y_{t,i}|M) \geq I(Y_{t,1}^{i-1};Y_{t,i}|M), ~ 1 \leq i \leq n.$
\item $ I(Y_{s,1}^{i-1};Y_{s,i}|M) \geq I(Y_{t,1}^{i-1};Y_{s,i}|M), ~ 1 \leq i \leq n.$
\end{enumerate}
\end{lemma}
\begin{proof}
The proof of Part 1 follows by progressively flipping one co-ordinate of $Y_{s1}^{i-1}$ to $Y_{t1}^{i-1}$,
and showing that the inequality holds at each flip using the less-noisy ($Y_s \succeq Y_t$) assumption.

Observe that for any $1 \leq r \leq i-1$
\begin{align*}
& I(Y_{t,1}^{r-1},Y_{s,r}^{i-1};Y_{ti}|M) \\
&\quad = I(Y_{t,1}^{r-1} ,Y_{s,r+1}^{i-1};Y_{t,i}|M) + I(Y_{s,r};Y_{t,i}|M,Y_{t,1}^{r-1} ,Y_{s,r+1}^{i-1}) \\
& \quad \stackrel{(a)}{\geq}  I(Y_{t,1}^{r-1} ,Y_{s,r+1}^{i-1};Y_{t,i}|M)  + I(Y_{t,r};Y_{t,i}|M,Y_{t,1}^{r-1} ,Y_{s,r+1}^{i-1}) \\
& \quad = I(Y_{t,1}^{r},Y_{s,r+1}^{i-1};Y_{ti}|M),
\end{align*}
where $(a)$  follows from the following  two observations:
\begin{itemize}
\item $(M, Y_{t,1}^{r-1} ,Y_{s,r+1}^{i-1} ,Y_{ti}) \to X_r \to (Y_{s,r}, Y_{t,r})$ forms a Markov chain
\item The receiver $Y_s$ is less noisy than $Y_{t}$ implying, in particular, that
$$ I(Y_{s,r};Y_{t,i}|M,Y_{t,1}^{r-1} ,Y_{s,r+1}^{i-1}) \geq I(Y_{t,r};Y_{t,i}|M,Y_{t,1}^{r-1} ,Y_{s,r+1}^{i-1}). $$
\end{itemize}

This yields us a chain of inequalities of the form
\begin{align*}
&  I(Y_{s,1}^{i-1};Y_{t,i}|M) \geq I(Y_{t,1}, Y_{s,2}^{i-1};Y_{ti}|M) \geq \cdots \\
& \quad \cdots \geq I(Y_{t,1}^{i-2}, Y_{s,i-1};Y_{ti}|M) \geq I(Y_{t,1}^{i-1};Y_{t,i}|M),
\end{align*}
thus establishing the Part 1 of the Lemma.

The proof of Part 2 follows identical arguments (replace $Y_{ti}$ by $Y_{si}$) as in the proof of Part 1 and is omitted.
\end{proof}

The main converse  follows using Fano's inequality and the above lemma.

Observe that
\begin{align*}
n R_3 &\leq I(M_3;Y_{3,1}^n) + n \epsilon_n \\
& = \sum_{i=1}^n I(M_3;Y_{3,i}| Y_{3,1}^{i-1}) + n \epsilon_n \\
& \leq \sum_{i=1}^n I(M_3,  Y_{3,1}^{i-1};Y_{3,i}) + n \epsilon_n \\
& \stackrel{(a)}{\leq} \sum_{i=1}^n I(M_3,  Y_{2,1}^{i-1};Y_{3,i}) + n \epsilon_n \\
& =  \sum_{i=1}^n I(U_i;Y_{3,i}) + n \epsilon_n,
\end{align*}
where $U_i = (M_3,  Y_{2,1}^{i-1})$. Here $(a)$ follows from Lemma \ref{le:ord}.

From Fano's inequality we also have
\begin{align*}
n R_2 &\leq I(M_2;Y_{2,1}^n|M_3) + n \epsilon_n \\
& = \sum_{i=1}^n I(M_2;Y_{2,i}| M_3, Y_{2,1}^{i-1}) + n \epsilon_n \\
& = \sum_{i=1}^n I(V_i;Y_{2,i}|U_i) + n \epsilon_n,
\end{align*}
where $V_i =  (M_2,M_3,  Y_{2,1}^{i-1})$.

Finally observe that
\begin{align*}
n R_1 &\leq I(M_1;Y_{1,1}^n|M_2,M_3) + n \epsilon_n \\
& = \sum_{i=1}^n I(M_1;Y_{1,i}| M_2, M_3 Y_{1,1}^{i-1}) + n \epsilon_n \\
& \stackrel{(a)}{\leq} \sum_{i=1}^n I(X_i;Y_{1,i}|M_2, M_3, Y_{1,1}^{i-1}) + n \epsilon_n \\
& \stackrel{(b)}{=} \sum_{i=1}^n I(X_i;Y_{1,i}|M_2, M_3)  - I(Y_{1,1}^{i-1};Y_{1,i}|M_2, M_3) +  n \epsilon_n \\
& \stackrel{(c)}{\leq} \sum_{i=1}^n I(X_i;Y_{1,i}|M_2, M_3)  - I(Y_{2,1}^{i-1};Y_{1,i}|M_2, M_3) +  n \epsilon_n 
\end{align*}
\begin{align*}
~~~& \stackrel{(d)}{\leq} \sum_{i=1}^n I(X_i;Y_{1,i}|M_2, M_3,Y_{2,1}^{i-1})  + \epsilon_n \\
& = \sum_{i=1}^n I(X_i;Y_{1,i}|V_i) + n \epsilon_n.
\end{align*}
Here $(a), (b),$ and $(d)$ follow from the data processing inequality and that
$$ (M_1,M_2,M_3, Y_{1,1}^{i-1}, Y_{2,1}^{i-1}) \to X_i \to Y_{1i}$$
forms a Markov chain. The inequality $(c)$ follows from Part 2 of Lemma \ref{le:ord}.

let $Q \in \{1,2,...,n\}$ to be a uniformly distributed random variable independent of all
other random variables. Setting $U = (U_Q,Q), V=(V_Q,Q), X=X_Q$ completes the converse in the standard way.
Clearly $U \to V \to X$ forms a Markov chain  as $V_i = (U_i,M_2)$. 

The cardinality arguments are  standard in literature (see \cite{czk78}, \cite{nae09}), and follows using the Fenchel-Eggleston strengthening of the usual Caratheodory's argument.

This completes the proof of the converse. \qed

A natural question here is whether this approach generalizes to more than three receivers. It appears to the authors that to generalize this argument to more than three receivers, one has to impose additional constraints on the class of $k$-receiver less broadcast noisy channels. Since this generalization leads to a rather interesting condition we shall define the class, and give a brief outline as to why the proof generalizes naturally under this setting.

\section{The $k$-receiver  interleavable broadcast channel}

\begin{definition}
A $k$-receiver less noisy broadcast channel is said to belong to be an {\em interleavable} broadcast channel if
there exists $k-1$ virtual receivers $V_1,...,V_{k-1}$ satisfying:
\begin{itemize}
\item $X \to V_1 \to ... \to V_{k-1}$ forms a Markov chain and
\item  The following``interleaved" less noisy condition holds:
\begin{equation}
Y_1 \succeq V_1 \succeq Y_2  \succeq \cdots Y_{k-1} \succeq V_{k-1}
\succeq Y_k. \label{eq:lnmod}
\end{equation}
\end{itemize}
\end{definition}

This class generalizes the 3-receiver less noisy broadcast channel. Indeed, the following broadcast channels are
some examples belonging to this class :
\begin{enumerate}
\item A sequence of degraded receivers, i.e. $X \to Y_1 \to ... \to Y_{k}$; set $V_i = Y_{i+1}$,
\item A sequence of "nested" less noisy receivers, i.e. $Y_s \succeq (Y_{s+1},...,Y_k)$; set $V_i = (Y_{i+1},...,Y_k),$
\item A 3-receiver less noisy sequence, i.e. $Y_1 \succeq Y_2 \succeq Y_3$; set $V_1=V_2=Y_2$.
\end{enumerate}

From Remark \ref{re:om} we know that the less-noisy ordering only depends on the marginals. Hence without loss of generality we can assume that the probability distribution factorizes as follows:
\begin{align*} 
& p (x^n, y_1^n , \ldots, y_k^n, v_1^n, \ldots , v_{k-1}^n) \\
&\quad = \prod_{i=1}^n p(x_i|x^{i-1}) p(y_{1i},..,y_{ki},v_{1i},..,v_{k-1,i}|x_i) \\
&\quad  =  \prod_{i=1}^n p(x_i|x^{i-1}) p(y_{1i},..,y_{ki}|x_i) p(v_{1i},..,v_{k-1,i}|x_i) \\
& \quad =  \prod_{i=1}^n p(x_i|x^{i-1}) p(y_{1i},..,y_{ki}|x_i) p(v_{1i}|x_{i}) \prod_{j=2}^{k-1} p(v_{ji}|v_{j-1,i})
\end{align*}

Here the first equality is due to the fact that the channel is DMC without feedback, second  is due to the fact that the assumptions on the less noisy structure just depends on the marginals, and third is due to the Markov chain $X \to V_1 \to ... \to V_{k-1}$. 

Given this structure we immediately observe the following Markov chain
\begin{align}
\label{eq:mc}
& V_{s,1}^{i-1} \to V_{s-1,1}^{i-1} \to X^n,  Y_1^n , \ldots, Y_k^n, M_1,...,M_k.
\end{align}
for $1 \leq s \leq k-1$; (set  $V_0 = X$). 

\begin{theorem}
The capacity region of a k-receiver interleavable less-noisy discrete memoryless broadcast channel is given by the union of rate triples $(R_1, \ldots, R_k)$ satisfying
\begin{align*}
R_l &\leq I(U_l;Y_l|U_{l+1}), ~ 1 \leq l \leq k,
\end{align*}
over all choices of $(U_2,...U_k,X)$ such that $(U_{k+1}=\emptyset) \to U_k \to \cdots U_2 \to (U_1 = X) \to (Y_1,Y_2, \ldots, Y_k)$ forms a  Markov chain. Further it suffices to consider $|U_{k-r}| \leq (|X| + 1)^{r+1}, ~ 1 \leq r \leq k-2$.
\end{theorem}
\begin{proof}
The proof is almost identical to that of the three receiver broadcast channel. The achievability proof is standard using superposition encoding and successive decoding and is omitted here. 

Let $M_{l+1}^k = (M_{l+1},...,M_k)$. Using Fano's
inequality, observe that for $1 \leq l \leq k$.
\begin{align*}
n R_l &\leq I(M_l;Y_{l,1}^n | M_{l+1}^k) + n \epsilon_n \\
&= \sum_{i=1}^n I(M_l;Y_{l,i}|M_{l+1}^k, Y_{l,1}^{i-1}) + n \epsilon_n \\
& = \sum_{i=1}^n I(M_l, Y_{l,1}^{i-1};Y_{l,i}|M_{l+1}^k) \\
& \quad - I(Y_{l,1}^{i-1};Y_{l,i}|M_{l+1}^k) + \epsilon_n \\
& \stackrel{(a)}{\leq}  I(M_l, Y_{l,1}^{i-1};Y_{l,i}|M_{l+1}^k) \\
& \quad - I(V_{l,1}^{i-1};Y_{l,i}|M_{l+1}^k) + \epsilon_n \\
& \stackrel{b)}{\leq}  I(M_l, V_{l-1,1}^{i-1};Y_{l,i}|M_{l+1}^k) \\
&\quad - I(V_{l,1}^{i-1};Y_{l,i}|M_{l+1}^k) + \epsilon_n \\
& \stackrel{(c)}{=}  I(M_l, V_{l-1,1}^{i-1};Y_{l,i}|M_{l+1}^k,V_{l,1}^{i-1} ) \epsilon_n \\
& = \sum_{i=1}^n I(U_{l,i};Y_{l,i}|U_{l+1, i}) + n \epsilon_n,
\end{align*}
where $U_{l,i} = (M_l^k,V_{l-1,1}^{i-1})$. We set $V_0 = X$. Here the inequalities $(a), (b)$ follow from
the Lemma \ref{le:ord} and that $V_{l-1} \succeq Y_l \succeq V_{s-1}$. The equality $(c)$ follows
from the Markov chain in \eqref{eq:mc}.

Define $Q$ to be a uniform random variable taking values in
$\{1,..,n\}$ and independent of all other random variables. As
usual, we set $U_l = (U_{l,Q},Q) ~\mbox{and}~ X=X_Q$. Since $X \to
V_1 \to \cdots \to V_{k-1}$ is a Markov chain it follows that $U_k
\to U_{k-1} \to \cdots \to U_2 \to X$ forms a Markov chain as well.
 The cardinality arguments are again standard and omitted. This completes the proof of the converse.
\end{proof}

\begin{remark}
It is not very difficult to observe that in general the 4-receiver less noisy broadcast channel is not an {\em interleavable} broadcast channel. To observe this let $Z_1 \succeq Z_2$ be any pair of less noisy but not degraded (stochastically) receivers. (Such a pair exists, see \cite{kom75} or \cite{nai08b}). Now let $Y_1,Y_2 \approx Z_1$ thus sandwiching $V_1 = Z_1$ and $Y_3,Y_4 \approx Z_2$ thus sandwiching $V_3 = Z_2$. However $X \to V_1 \to V_3$ cannot be a Markov chain by the assumption on $Z_1, Z_2$. Hence the problem of determining the capacity of $k$-receiver less noisy channel $k \geq 4$ is still very much open.
\end{remark}

\section{Conclusion}
We establish the capacity region for the 3-receiver less noisy
broadcast channel. We also compute the capacity region for a class
of k-receiver less noisy sequences that contain the previously
mentioned scenario. A new information inequality
is used to obtain the converse. and this technique also  simplifies the original proof \cite{kom75} of 
the converse of the 2-receiver broadcast channel.

\newpage

\section*{Acknowledgements}
The authors are very grateful to Prof. Abbas El Gamal for suggestions on improving the presentation of this result.

\bibliographystyle{amsplain}
\bibliography{mybiblio}

\end{document}